\documentclass[letterpaper, 10 pt, conference]{ieeeconf} 

\IEEEoverridecommandlockouts                            
                                                          
\usepackage[utf8]{inputenc}
\usepackage[T1]{fontenc}
\usepackage{amsmath,amssymb,amsfonts}
\usepackage{times}
\usepackage{graphicx}
\usepackage{subcaption}
\usepackage{epsfig}
\usepackage{cite}
\usepackage{hyperref}
\usepackage{color}
\usepackage{mathtools}
\usepackage{bm}
\usepackage{algorithm}
\usepackage{algpseudocode}
\usepackage{url}
\usepackage{tikz}
\usepackage[percent]{overpic}
\usetikzlibrary{positioning, arrows.meta}
\usepackage{xcolor}

\newtheorem{lemma}{\textbf{Lemma}}
\newtheorem{theorem}{\textbf{Theorem}}

\title{\LARGE \bf
Pitot-Aided Attitude and Air Velocity Estimation with Almost Global Asymptotic Stability Guarantees
}

\author{Méloné Nyoba Tchonkeu$^{1}$, Soulaimane Berkane$^{2}$, and Tarek Hamel$^{3}$ 
\thanks{This work was supported by the"Grands Fonds Marins" Project Deep-C, and the ASTRID ANR project ASCAR. This research work is also supported in part by NSERC-DG RGPIN-2020-04759 and Fonds de recherche du Québec (FRQ).}%
\thanks{$^{1}$M. Nyoba Tchonkeu is with the Department of Computer Science and Engineering, University of Quebec in Outaouais, 
        Gatineau, QC J8X3X7, Canada
        {\tt\small (nyom01@uqo.ca)}}%
\thanks{$^{2}$S. Berkane is with the Department of Computer Science and Engineering, University of Quebec in Outaouais, Gatineau, QC J8X3X7, and also with the Department of Electrical Engineering, Lakehead University,
        Thunder Bay, ON P7B 5E1, Canada
        {\tt\small (soulaimane.berkane@uqo.ca)}}%
\thanks{$^{3}$T. Hamel is with I3S-UniCA-CNRS, University Cote d’Azur and the Insitut Universitaire de France,
        06903 Sophia Antipolis, France
        {\tt\small (thamel@i3s.unice.fr)}}%
}

\begin{document}

\maketitle
\thispagestyle{empty}
\pagestyle{empty}

\begin{abstract}
This paper investigates the problem of attitude and air velocity estimation for fixed-wing unmanned aerial vehicles (UAVs) using IMU measurements and at least one Pitot tube measurement, with almost global asymptotic stability (AGAS) guarantees. A cascade observer architecture is developed, in which a Riccati/Kalman-type filter estimates the body-fixed frame air velocity and the vehicle's tilt using IMU data as inputs and Pitot measurements as outputs. Under mild excitation conditions, the resulting air velocity and tilt estimation error dynamics are shown to be uniformly observable. The estimated tilt is then combined with magnetometer measurements in a nonlinear observer on $\mathrm{SO}(3)$ to recover the full attitude. Rigorous analysis establishes AGAS of the overall cascade structure under the uniform observability (UO) condition. The effectiveness of the proposed approach is demonstrated through validation on real flight data.
\end{abstract}

\section{INTRODUCTION}
Reliable attitude estimation is a fundamental requirement for UAVs, particularly in GNSS-denied or highly dynamic environments, where navigation performance relies primarily on onboard sensing. In such conditions, the inertial measurement unit (IMU), providing angular velocity and specific acceleration measurements, serves as the primary sensing modality. In many practical estimation frameworks, particularly for lightweight UAVs, the gravity direction is inferred from accelerometer measurements under the assumption that non-gravitational accelerations are negligible \cite{mahony2008nonlinear,hua2013implementation}. Since accelerometers measure specific force rather than gravity, this approximation is only valid when translational accelerations remain small \cite{mahony2008nonlinear}; during sustained forward flight or aggressive maneuvers, tilt estimation based solely on accelerometer information becomes unreliable, necessitating additional measurements to restore attitude observability.

A common approach to overcoming this limitation is to augment inertial sensing with velocity-related information, giving rise to the class of velocity-aided attitude (VAA) estimation problems. Velocity measurements introduce a coupling between translational and rotational dynamics that can be exploited for attitude estimation. Numerous observer designs have been proposed, including linearized designs and invariant observer formulations \cite{2008_bonnabel_SymmetryPreservingObservers,2008_martin_InvariantObserverEarthVelocityAided}, as well as nonlinear observers with strong stability properties, in some cases achieving almost global asymptotic stability guarantees under suitable observability conditions \cite{2013_troni_PreliminaryExperimentalEvaluation,hua2010attitude,roberts2011attitude,berkane2017attitude,2016_hua_StabilityAnalysisVelocityaided,wang2021nonlinear,benallegue2023velocity,Pieter2023}. However, most VAA approaches assume the availability of full velocity measurements expressed either in body-fixed frame or in inertial frame. This assumption is often restrictive for small UAV platforms. In fixed-wing UAVs, air velocity is typically measured using Pitot tubes and expressed in body-fixed frame. Air velocity measurements are intrinsically coupled with the vehicle attitude through the translation dynamics and can therefore contribute to attitude observability. Early approaches combined Pitot and IMU data with aerodynamic models \cite{lie2014synthetic, borup2016nonlinear}, but their effectiveness depended on poorly known aerodynamic parameters. More recent works incorporate Pitot measurements directly into observer designs \cite{sun2019observability,johansen2015estimation}. In particular, Oliveira et al.~\cite{oliveira2024pitot} proposed a Riccati observer framework on $\mathrm{SO}(3) \times \mathbb{R}^3$ combining IMU data with a \textit{single-axis} Pitot tube to estimate the body-fixed frame air velocity and the vehicle tilt under suitable excitation conditions. Owing to the axial symmetry of the Pitot measurement, the yaw angle remains unobservable  in the absence of an inertial heading reference, restricting the estimation to a reduced attitude. Moreover, the nonlinear observer structure guarantees only local convergence.

The present paper addresses these limitations through a cascade observer design that separates aerodynamic-aided tilt estimation from full attitude estimation. In the first stage, we design a Riccati/Kalman-type observer that jointly estimates air velocity and the vehicle's tilt, defined as the gravity direction in the body-fixed frame. Under appropriate uniform excitation conditions, the resulting estimation error dynamics are shown to be globally uniformly exponentially stable. In the second stage, the estimated tilt is fused with magnetometer measurements in a nonlinear attitude filter on $\mathrm{SO}(3)$ based on the geometric filter introduced in \cite{tchonkeu2025barometer}. The magnetometer provides an independent directional reference that resolves the yaw ambiguity inherent to Pitot-based sensing, enabling estimation of the full attitude with almost global asymptotic stability guarantees. This approach extends existing Pitot-assisted estimation framework beyond local convergence results by providing a unified framework for full air velocity and attitude estimation with rigorous stability guarantees. Experimental results demonstrate the effectiveness of the proposed estimation scheme.  

\section{PRELIMINARY MATERIAL}
\label{sec:prelims}

We denote by \( \mathbb{R} \) and \( \mathbb{R}_+ \) the sets of real and nonnegative real numbers, respectively. The \( n \)-dimensional Euclidean space is denoted by \( \mathbb{R}^n \). The Euclidean inner product of two vectors \( x, y \in \mathbb{R}^n \) is defined as \( \langle x, y \rangle = x^\top y \). The associated Euclidean norm of a vector \( x \in \mathbb{R}^n \) is \( |x| = \sqrt{x^\top x} \). Furthermore, we denote by \( \mathbb{R}^{m \times n} \) the set of real \( m \times n \) matrices. The set of \( n \times n \) positive definite matrices is denoted by \( \mathcal{S}^+(n) \), and the identity matrix is denoted by \( I_n \in \mathbb{R}^{n \times n} \). Given two matrices \( A, B \in \mathbb{R}^{m \times n} \), the Euclidean matrix inner product is defined as \( \langle A, B \rangle = \mathrm{tr}(A^\top B) \), and the Frobenius norm of \( A \in \mathbb{R}^{n \times n} \) is given by \( \|A\| = \sqrt{\langle A, A \rangle} \). The unit sphere \( \mathbb{S}^2 := \{ \eta \in \mathbb{R}^3 \mid |\eta| = 1 \} \subset \mathbb{R}^3 \) denotes the set of unit 3D vectors and forms a smooth submanifold of \( \mathbb{R}^3 \).

The special orthogonal group of 3D rotations is denoted by
$
\mathrm{SO}(3) := \{ R \in \mathbb{R}^{3 \times 3} \mid RR^\top = R^\top R = I_3,\ \det(R) = 1 \}.
$
The Lie algebra of $\mathrm{SO}(3)$ is 
$
\mathfrak{so}(3) := \{\, \Omega \in \mathbb{R}^{3\times 3} \mid \Omega^\top = -\Omega \,\},
$
isomorphic to $\mathbb{R}^3$ via the skew-symmetric operator 
$(\cdot)^\times : \mathbb{R}^3 \to \mathfrak{so}(3)$, defined such that 
$
u \times v = u^\times v$ for all $u,v \in \mathbb{R}^3.$
The exponential map \( \exp : \mathfrak{so}(3) \rightarrow \mathrm{SO}(3) \) defines a local diffeomorphism from a neighborhood of \( 0 \in \mathfrak{so}(3) \) to a neighborhood of \( I_3 \in \mathrm{SO}(3) \). This enables the composition map \( \exp \circ (\cdot)^\times : \mathbb{R}^3 \rightarrow \mathrm{SO}(3) \), which is given by the following Rodrigues' formula~\cite{Ma2004rodriguesformula} :

\begin{equation}
\exp([\theta]^{^\times}) 
= I_3 - \frac{\sin(\|\theta\|)}{\|\theta\|}[\theta]^{^\times}
+ \frac{1-\cos(\|\theta\|)}{\|\theta\|^2}([\theta]^{^\times})^2.
\label{eq:rodriguesformula}
\end{equation}

Consider the linear time-varying (LTV) system given by
\begin{equation}
\begin{cases}
\dot{x} = A(t)x + B(t)u, \\
y = C(t)x,
\end{cases}
\label{eq:LTV_system}
\end{equation}
with state \( x \in \mathbb{R}^n \), input \( u \in \mathbb{R}^\ell \), and output \( y \in \mathbb{R}^m \). The matrix-valued functions \( A(t) \), \( B(t) \), and \( C(t) \) are assumed to be continuous and bounded. By definition from~\cite{Besancon2007}, the system ~\eqref{eq:LTV_system} or pair \( (A(t), C(t)) \) is \emph{uniformly observable} if there exist constants 
\( \delta, \mu > 0 \) such that, for all \( t \geq 0 \),
\begin{equation}
W(t, t+\delta) := \frac{1}{\delta} \int_t^{t+\delta} 
\Phi^\top(s, t) \, C^\top(s) \, C(s) \, \Phi(s, t) \, ds 
\geq \mu I_n, \label{eq:observability_gramian}
\end{equation}
where \( \Phi(s, t) \) is the state transition matrix such that
\begin{equation}
\frac{d}{dt} \Phi(s, t) = A(t)\Phi(s, t), 
\quad \Phi(t, t) = I_n, \quad \forall s \geq t. \label{eq:transition_matrix} 
\end{equation}
\( W(t, t+\delta) \) is called the \emph{observability Gramian} of the system.

\section{PROBLEM DESCRIPTION}
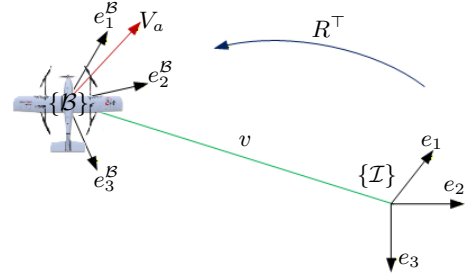
\begin{figure}[ht]
  \centering
  \begin{overpic}[width=0.7\linewidth]{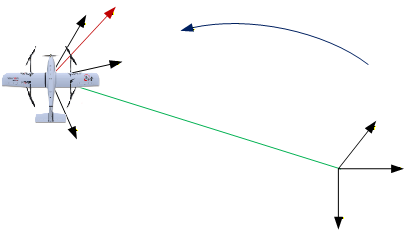}
    \put(8,36){\small \(\left\{\mathcal{B}\right\}\)}
    \put(18,55){\small \(e_1^{\mathcal{B}}\)}
    \put(30,42){\small \(e_2^{\mathcal{B}}\)}
    \put(18,19){\small \(e_3^{\mathcal{B}}\)}
    \put(76,20){\small \(\left\{\mathcal{I}\right\}\)}
    \put(90,28){\small \(e_1\)}
    \put(95,18){\small \(e_2\)}
    \put(85,2){\small \(e_3\)}
    \put(50,28){\small \(v\)}
    \put(66,52){\small \(R^\top\)}
    \put(28,54){\small \(V_a\)}
  \end{overpic}
  \caption{A fixed-wing UAV equipped with an IMU, a Pitot tube measuring the air velocity $V_a$.}
  \label{fig:veh_pitot_imu}
\end{figure}
Consider a vehicle equipped with an IMU, moving in 3D space. Let $\mathcal{I}$ denote a North-East-Down (NED) inertial frame, and let $\mathcal{B}$ be a body-fixed NED frame attached to the vehicle's center of mass and colocated with the IMU frame.
Let \( R \in \mathrm{SO}(3) \) denote the rotation matrix from the vehicle body-fixed frame \( \mathcal{B} \) to the inertial frame \( \mathcal{I} \), and let \( \boldsymbol{\omega} \in \mathbb{R}^3 \) and  \( \mathbf{a} \in \mathbb{R}^3 \) be the body angular velocity and the linear specific acceleration measured by the IMU, respectively. The linear velocity of the rigid body, expressed in the inertial frame \( \mathcal{I} \), is denoted by \( v \in \mathbb{R}^3 \). The gravitational acceleration is expressed in the inertial frame by \( \mathbf{g} = g \mathbf{e}_3 \in \mathbb{R}^3 \), where \( \mathbf{e}_3 = \begin{bmatrix}0  & 0 & 1 \end{bmatrix}^\top \in \mathbb{S}^2 \) denotes the standard gravity direction, and \( g \approx 9.81\, \text{m/s}^2 \) is the gravity constant. The vehicle's translational and rotational kinematics are given by
\begin{align*}
\dot{v} &= R\mathbf{a} + g \mathbf{e}_3, \\
\dot{R} &= R\boldsymbol{\omega}^\times.
\end{align*}
Let \( v_w \in \mathbb{R}^3 \) be the inertial wind velocity. The vehicle air velocity in inertial and body-fixed frames are \(v_a = v - v_w \) and \(V_a = R^{\top}v_a\), respectively. Assuming a bounded and slowly time-varying wind velocity throughout the flight ( \(\dot v_w \approx 0 \)), one obtains \(\dot v_a = \dot v\). Hence, the body-fixed frame air velocity and attitude dynamics become :
\begin{align}
\dot{V}_a &= -\boldsymbol{\omega}^\times V_a +gR^\top \mathbf{e}_3 + \mathbf{a}, \label{eq:air_velocity_dyn} \\
\dot{R} &= R\boldsymbol{\omega}^\times.
\label{eq:attitude_dyn}
\end{align}

The vehicle is equipped with $m$ calibrated Pitot probes. Each probe measures the scalar projection of the air velocity $V_a = \begin{bmatrix}V_{a,1} & V_{a,2} & V_{a,3}\end{bmatrix}^{\top}$ along a known body-fixed direction \(b_i \in \mathbb{S}^2\):
\begin{equation} 
y_{p,i} = b_i^\top V_a,
\end{equation}
Collectively, multi-probe Pitot system vector is modeled as
\begin{equation}
y_p = B^\top V_a, \quad B = [b_1\ \cdots b_m] \in \mathbb{R}^{3\times m}, \label{eq:pitot_measurement}
\end{equation}
For small fixed-wing UAVs, a single Ptiot probe aligned with the longitudinal body-axis is typically used, corresponding to $m = 1$ and $B=\mathbf{e}_1$, with \( \mathbf{e}_1 = \begin{bmatrix}1  & 0 & 0 \end{bmatrix}^\top \in \mathbb{S}^2 \).
This aerodynamic measurement provides only partial information on the air velocity. However, by exploiting the system dynamics (specifically the coupling between the air velocity and the attitude \refeq{eq:air_velocity_dyn}-\refeq{eq:attitude_dyn}), we will show that it can be sufficient to estimate both the full air velocity vector and the vehicle reduced attitude (tilt) even in the absence of GNSS velocity data.
In addition, we assume available a magnetometer that provides measurements of the Earth's magnetic field,
\begin{equation}
\mathbf{m}_{\mathcal{B}} = R^\top \mathbf{m}_{\mathcal{I}},
\label{eq:mag_measurement}
\end{equation}
where $\mathbf{m}_{\mathcal{I}} \in \mathbb{S}^2$ is the known magnetic field vector expressed in the inertial frame. 

The objective is to estimate the attitude $R \in \mathrm{SO}(3)$ and the air velocity $V_a \in \mathbb{R}^3 $ using airspeed and heading information, while exploiting IMU measurements through the coupled dynamics~\eqref{eq:air_velocity_dyn} and~\eqref{eq:attitude_dyn}.

\section{PROPOSED OBSERVER DESIGN}

This section presents a two-stage observer architecture for attitude and air velocity estimation. The design exploits a reduced-order Riccati observer for air velocity dynamics and tilt estimation, followed by a nonlinear observer on \( \mathrm{SO}(3) \) to estimate full orientation.  The overall architecture is illustrated in figure~\ref{fig:observer_block_diagram}.

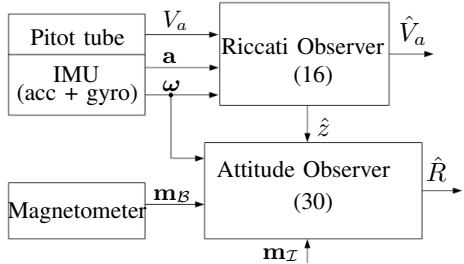
\begin{figure}[ht]
  \centering
  \begin{overpic}[width=0.7\linewidth]{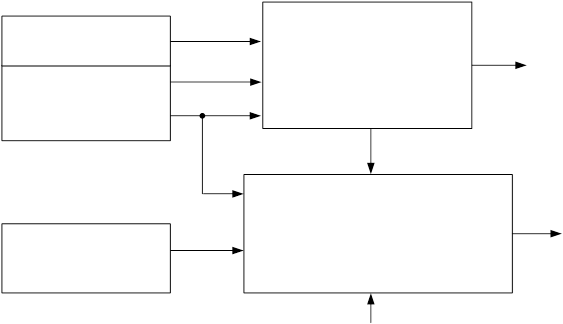}
    \put(5,48){\small Pitot tube}
    \put(10,40){\small IMU}
    \put(2,36){\small (acc + gyro)}
    \put(0,10){{\small Magnetometer}}
    \put(34,52){\small \(V_a\)}
    \put(34,44){\small \(\mathbf{a}\)}
    \put(34,38){\small \(\boldsymbol{\omega}\)}
    \put(32,14){\small \(\mathbf{m}_{\mathcal{B}}\)}
    \put(47,46){\small Riccati Observer}
    \put(61,40){\small~\eqref{eq:riccati_observer}} 
    \put(86,49){\normalsize \(\hat{V}_a\)}
    \put(68,28){\normalsize \(\hat{z}\)}
    \put(46,19){\small Attitude Observer}
    \put(61,12){\small~\eqref{eq:attitude_observer}}
    \put(92,18){\normalsize \(\hat{R}\)}
    \put(56,1){\small \(\mathbf{m}_{\mathcal{I}}\)}
  \end{overpic}
  \caption{Two-stage observer architecture: a Riccati observer fuses IMU and Pitot tube data to estimate tilt and air velocity, followed by an \( \mathrm{SO}(3) \) observer combining tilt, gyro, and magnetometer to estimate \( \hat{R} \).}
  \label{fig:observer_block_diagram}
\end{figure}

\subsection{Tilt and Air Velocity Estimation}

We define the tilt (reduced attitude) as the gravity direction expressed in the body frame~\cite{benallegue2023velocity}:
\begin{equation}
z := R^\top \mathbf{e}_3 \in \mathbb{S}^2,
\end{equation}
which evolves according to
\begin{equation}
\dot{z} = -\boldsymbol{\omega}^\times z.
\label{eq:z_dot}
\end{equation}
Combining \eqref{eq:air_velocity_dyn} and \eqref{eq:z_dot}, yields the system dynamics
\begin{equation}
\begin{cases}
\dot{V}_a = -\boldsymbol{\omega}^\times V_a +gz + \mathbf{a},\\
\dot{z} = -\boldsymbol{\omega}^\times z. \label{eq:system_dynamics}
\end{cases}
\end{equation}
By defining the augmented state $z \in \mathbb{R}^3$, the system dynamics~\eqref{eq:system_dynamics} and the Pitot tube measurement~\eqref{eq:pitot_measurement} can be rewritten as a linear time-varying (LTV) system.
Letting $x := \begin{bmatrix} V_a^\top & z^\top \end{bmatrix}^\top \in \mathbb{R}^6$, we obtain
\begin{equation}
\begin{cases}
\dot{x} = A(t)x + B_u\mathbf{a}, \\
y_p = Cx,
\end{cases}
\label{eq:ltv}
\end{equation}
with
\begin{equation}
A(t) = 
\begin{bmatrix}
-\boldsymbol{\omega}^\times  & gI_3 \\
\mathbf{0}_{3 \times 3} & -\boldsymbol{\omega}^\times \\
\end{bmatrix}= A_w(t) + \bar{A}, \quad 
B_u = \begin{bmatrix} I_3 \\ \mathbf{0}_{3 \times 3} \end{bmatrix}, \label{eq:state_command_matrix}
\end{equation}
and the output matrix :
\begin{equation}
C = \begin{bmatrix} B^{\top} & \mathbf{0}_{m \times 3} \end{bmatrix} = B^{\top}\bar{C} \in \mathbb{R}^{m\times 6}, \label{eq:output_matrix}
\end{equation}
with
\[
A_w(t) =\begin{bmatrix}-\boldsymbol{\omega}^\times & \mathbf{0}_{3 \times 3} \\
\mathbf{0}_{3 \times 3} & -\boldsymbol{\omega}^\times
\end{bmatrix}, \quad \bar{A} = \begin{bmatrix} \mathbf{0}_{3 \times 3} & gI_3 \\ \mathbf{0}_{3 \times 3}  & \mathbf{0}_{3 \times 3}  \end{bmatrix},
\] and \[\bar{C} = \begin{bmatrix} I_3 & \mathbf{0}_{3 \times 3} \end{bmatrix} \in \mathbb{R}^{3\times 6}.\]
Using the measured angular velocity, the state vector $x$ is estimated without explicitly reconstructing the attitude $R$, by decoupling the state dynamics from the attitude estimation and designing a Riccati observer of the Kalman–Bucy type. 
The observer is given by
\begin{equation}
\dot{\hat{x}} = A(t) \hat{x} + B_u\mathbf{a}  + K(t)\left(y_p - C \hat{x} \right), \label{eq:riccati_observer}
\end{equation}
where $K(t) = P(t) C^\top Q^{-1}$ and and $P(t)$ solves the continuous-time Riccati equation (CRE): \(\dot{P} = A(t) P + P A^\top(t) - P C^\top Q^{-1} C P + S\), with \(P(0) > 0\). The matrices $Q>0$ and $S>0$ are symmetric positive definite, where $Q$ typically encodes the Pitot tube measurement covariance,  
while $S$ plays the role of the process noise covariance. The stability and convergence properties of this observer are fundamentally linked to the well-posedness of the CRE, and more precisely, to the UO of the LTV system \eqref{eq:ltv}. The UO property guarantees the existence of a unique, bounded, and positive-definite solution \( P(t) \) to the CRE for all \( t \ge 0 \). As a consequence, the estimation error \(\tilde{x} := \hat x - x\) between the observer's state and the actual state decays exponentially to zero with the rate of convergence tuned by \(Q\) and \(S\) \cite{Hamel2017PositionMeasurements}.

To analyze uniform observability, we proceed by deriving the observability Gramian. In view of the structure of the LTV system \eqref{eq:ltv}-\eqref{eq:state_command_matrix}, we consider the change of variable $\bar{x}(t) = T(t)x(t),$ with $T(t) = \mathrm{diag}\left (\phi(t),\phi(t)\right)$ where $\phi(t) \in \mathrm{SO}(3)$ and inherits the dynamics of $R(t)$. That is:
\begin{align}\label{eq:phi}
    \dot \phi &=\phi \boldsymbol{\omega}^\times ,\qquad \phi(0)=\phi_0\in\mathrm{SO}(3),
\end{align} 
Differentiating yields \[\dot {\bar{x}} := \dot{T}(t)x(t)+T(t)\big(\left(A_w(t) +\bar{A}\right)x(t)+B_u\mathbf{a}\big).\]
In view of \eqref{eq:attitude_dyn}, $\dot{T}(t) := T(t)\mathrm{diag}(\boldsymbol{\omega}^\times,\boldsymbol{\omega}^\times) = -T(t)A_w(t)$. Then, 
\[
\dot{\bar{x}} =T(t)\bar{A}x(t)+T(t)B_u\mathbf{a} = T(t)\bar{A}T(t)^{\top}\bar{x}(t)+T(t)B_u\mathbf{a}.
\] 
Since $\phi \phi^{\top}=I_3$, it follows that $T(t)\bar{A}T(t)^{\top} = \bar{A},$ and hence: 
\begin{equation}
\dot{\bar{x}}(t)  = \bar{A}\bar{x}(t)+T(t)B_u\mathbf{a}, \mbox{ and } y=C\phi(t)^\top\bar{x}. \label{eq:bar_x}
\end{equation}
Since the input $\mathbf{a}(t)$ does not affect the observability of the LTV system, we assume  $\mathbf{a}(t) = 0$ to simplify the analysis. In this case one deduces that:  $\bar{x}(t) = \exp{\left(\bar{A}(t-s)\right)}\bar{x}(s)$ for any $0\le s \le t$. Therefore, $x(t) = T(t)^{\top}\exp \left(\bar{A}(t-s)\right) T(s) x(s)$,
which implies that the state transition matrix is given by
\begin{equation}
\begin{aligned}
\Phi(t,s)&:=T^{\top}(t)\exp \left(\bar{A}(t-s)\right)T(s),\\
&= T(t)^{\top}\bar{\Phi}(t,s)T(s), \label{eq:state_trans_matrix}
\end{aligned}
\end{equation}
with $\bar{\Phi}(t,s) = \exp \left(\bar{A}(t-s)\right)$ and $T(s) = I_6.$
From \eqref{eq:state_trans_matrix}, we compute $C(s)\Phi(t,s)$ as follows
\begin{equation}
C(s)\Phi(t,s) = C(s)T^\top (t)\bar{\Phi}(t,s)T(s). \label{eq:C_Phi}
\end{equation}
In view of \eqref{eq:C_Phi} and \eqref{eq:observability_gramian}, the observability Gramian yields
\begin{equation}
\begin{aligned}
&W(t,t+\delta) =\\
&\frac{1}{\delta}\int_{t}^{t + \delta}T(t)^{\top}\bar{\Phi}^{\top}(s,t)T(s)C^{\top}(s)C(s)T(s)^{\top}\bar{\Phi}(s,t)T(t)ds.
\end{aligned}
\label{eq:gramian}
\end{equation}
Since $C(s) = B^{\top}\bar{C}$, we have

\begin{equation}
W(t,t+\delta) =
T(t)^{\top}\left(\frac{1}{\delta}\int_{t}^{t + \delta}\bar{\Phi}^{\top}(s,t)\Gamma(s)\bar{\Phi}(s,t)ds\right)T(t),
\label{eq:gramian_1}
\end{equation}
with $\Gamma(s) = T(s)\bar{C}^{\top}BB^{\top}\bar{C}T(s)^{\top}$. Moreover, we have
\begin{equation}
\bar{C}T(s)^{\top} = \bar{C}\begin{bmatrix}\phi(s)^{\top}  & \mathbf{0}_{3 \times 3} \\\mathbf{0}_{3 \times 3}&\phi(s)^{\top} \end{bmatrix} = \phi(s)^{\top}\bar{C}. 
\end{equation}
Thus the observability Gramian in \eqref{eq:gramian_1} becomes
\begin{equation}
W(t,t+\delta) =T(t)^{\top}\bar{W}(t,t+\delta)T(t)^\top 
\end{equation}
where 
\begin{equation}
\bar{W}(t,t+\delta)=\frac{1}{\delta}\int_{t}^{t + \delta}\bar{\Phi}^{\top}(s,t)\bar{C}^{\top}\Sigma(s)\bar{C}\bar{\Phi}(s,t)ds,
\label{eq:gramian_2}
\end{equation}
is the observability Gramian of the transformed system \eqref{eq:bar_x},
with $\Sigma(s) = \Pi(s)\Pi(s)^{\top}$, where $\Pi(s)= \phi(s)B.$

\begin{lemma} \label{Lemma1}
Assume that the angular velocity \(\boldsymbol{\omega}(t)\) is continuous and bounded. If the matrix \(\Sigma(t)\) is \emph{persistently exciting} (PE) in the sense that there exist  \( \bar{\delta}, \bar{\mu} > 0 \) such that \( \forall t \ge 0 \),
\begin{equation}
\frac{1}{\bar\delta}\int_t^{t+\bar\delta}\Sigma(s)\,ds
\ \ge\ \bar\mu I_3,
\label{eq:PE_Condition}
\end{equation}
then the pair $(A(t), C)$ is \emph{uniformly observable}, and the equilibrium \(\tilde x = 0_{6\times1}\) is globally uniformly exponentially stable (GES).
\end{lemma}

\begin{proof}
To establish uniform observability of the system~\eqref{eq:ltv}, it suffices to show that the transformed system is uniformly observable. We begin by examining the associated constant pair $(\bar{A}, \bar{C})$. From the structure of these matrices, it follows that $(\bar{A}, \bar{C})$ is Kalman observable and $\bar{A}$ has real eigenvalues. Hence, under the PE assumption in \eqref{eq:PE_Condition}, it follows from \cite[Lemma~2.7]{Hamel2017PositionMeasurements} that there  \( \delta, \mu > 0 \) such that
\[\bar{W}(t,t+\delta) \ge \mu I_6.\] This implies in view of \eqref{eq:gramian_2} and since $\|\bar{W}(t,t+\delta)\|=\|W(t,t+\delta)\|$ that \(W(t, t + \delta) \ge \mu I_6,~ \forall t \geq 0.\) Therefore, the pair \((A(t),C)\) is uniformly observable. This in turn guarantees the global exponential stability of the equilibrium \(\tilde{x} = \boldsymbol{0}_{6\times 1}\)(See~\cite{Hamel2017PositionMeasurements}).
\end{proof}
Note that, under the NED body-axis convention, the air velocity can be parametrized as 
\(V_a = \|V_a\|\begin{bmatrix}c_\alpha c_\beta ~~ c_\alpha s_\beta ~~ s_\alpha \end{bmatrix}^{\top}\)(see \cite{oliveira2024pitot}), where $\|V_a\|$ denotes the airspeed magnitude, and $\alpha$ and $\beta$ represent the angle of attack and sideslip angle, respectively. Accordingly, when the UO conditions of Lemma \ref{Lemma1} are fulfilled, the airspeed estimate is given by $\|\hat V_a\|$, while the angles are computed as \cite{oliveira2024pitot}
\begin{equation}
\hat \alpha = \arcsin\left(\frac{\hat{V}_{a,3}}{\|\hat V_a\|}\right),\quad \hat \beta = \arctan\left(\frac{\hat V_{a,2}}{\hat{V}_{a,1}}\right). \label{eq:attack_sideslip}
\end{equation}

For the standard single-Pitot configuration, the Pitot probe is aligned with the longitudinal body-axis, i.e. in view of \eqref{eq:pitot_measurement}, the Pitot measurement is given by 
\begin{equation}
y_p =\mathbf{e}_1^{\top}V_a = V_{a,1}, \label{eq:pitot_meas_conf_2}
\end{equation} 
With $m=1$ and $B = \mathbf{e}_1$, Lemma~\ref{Lemma1} requires sufficient excitation of the direction $\Pi(t) = \phi(t)\mathbf{e}_1$. In practice, $\phi(t)$ corresponds to the vehicle orientation, that is, $\phi(t) = R(t)$. Such excitation is typically induced by combined yaw and pitch motions, ensuring uniform observability of the full air velocity and the gravity direction. However, pitching a fixed-wing aircraft is more demanding and generally involves a time-varying angle of attack, which can be problematic. This structural limitation motivates the introduction of additional information to reduce the remaining rotational ambiguity. In particular, when physically justified, a zero-sideslip condition ($\beta = 0$) can be enforced and interpreted as a pseudo-measurement.

 In steady, coordinated flight, directional stability and coordinated control align the body with the relative wind, so that the mean sideslip angle is driven to zero. Under the parametrization in \eqref{eq:attack_sideslip}, this is equivalent to constraining $V_{a,2}=0$. Within the proposed framework, the pseudo-measurement is then incorporated as an additional output of the Riccati observer \eqref{eq:riccati_observer}, by augmenting the Pitot tube measurement as 
\begin{equation}
y_p = \begin{bmatrix}V_{a,1}& V_{a,2}\end{bmatrix}^{\top} =  \begin{bmatrix}V_{a,1}& \mathbf{0}\end{bmatrix}^{\top} ,\label{eq:pitot_meas_conf_1}
\end{equation}
which in view of \eqref{eq:pitot_measurement} corresponds to 
$m=2$ and $B = \begin{bmatrix} \mathbf{e}_1 & \mathbf{e}_2 \end{bmatrix}$.
With this output augmentation, the PE condition in \eqref{eq:PE_Condition} applies to the two-dimensional subspace spanned by the directions \(R(t)^\top\mathbf{e}_1\) and \(R(t)^\top\mathbf{e}_2\), requiring a sustained \textit{yaw motion only} to ensure persistent excitation, thereby relaxing the excitation requirement compared to the single-axis Pitot configuration. Moreover, this approach reduces the symmetry of the output map and strengthens the UO properties, since the lateral air velocity component is no longer free to evolve independently.

\subsection{Full Attitude Estimation on $ \mathrm{SO}(3)$}
Following the estimation of tilt vector \( \hat{z} \) via the Riccati observer, the full attitude estimation requires an additional known direction to resolve the heading ambiguity. To this end, we incorporate magnetometer measurements \(\mathbf{m}_{\mathcal{B}} \in \mathbb{S}^2 \) defined in \eqref{eq:mag_measurement}. Let \( \hat{R} \in \mathrm{SO}(3) \) denote the estimate of the true attitude \( R \). The attitude observer proposed in \cite{tchonkeu2025barometer} is given by
\begin{equation}
\dot{\hat{R}} = \hat{R} \boldsymbol{\omega}^\times - \sigma_R^\times \hat{R}, \quad \hat{R}(0) \in \mathrm{SO}(3),
\label{eq:attitude_observer}
\end{equation}
where the correction term \( \sigma_R \in \mathbb{R}^3 \) is defined as
\begin{equation}
\sigma_R = k_z (\mathbf{e}_3 \times \hat{R} \hat{z}) + k_m (\bar{\mathbf{m}}_{\mathcal{I}} \times \hat{R} \bar{\mathbf{m}}_{\mathcal{B}}),
\label{eq:attitude_correction}
\end{equation}
with \( \bar{\mathbf{m}}_{\mathcal{I}} = \bar\pi_{\mathbf{e}_3} \mathbf{m}_{\mathcal{I}} \), \( \bar{\mathbf{m}}_{\mathcal{B}} = \bar{\pi}_{\hat z} \mathbf{m}_{\mathcal{B}} \), \( k_z > 0 \), and \( k_m \geq 0 \); where for a vector $u \in \mathbb{R}^3,$ $\bar{\pi}_u := |u|^2 I_3 - uu^\top$ is a regularized projection operator, well defined even when $u=0_{3 \times 1}$.
The use of $\bar{\pi}_{\hat z}$ prevents singularities if $\hat z$ vanishes, and the projected magnetometer vectors ensure that yaw estimation is decoupled from roll and pitch; see \cite{hua2013implementation}. 

\begin{theorem}\label{Theorem1}
Consider the Riccati observer~\eqref{eq:riccati_observer} and the attitude observer~\eqref{eq:attitude_observer} with correction~\eqref{eq:attitude_correction}. Assume the pair \((A(t), C)\) is uniformly observable according to Lemma~\ref{Lemma1}, then the estimation error \(\tilde{R}=\hat RR^\top, \tilde{x}=x-\hat x\) converge to the set of equilibria \(\mathcal{E} = \mathcal{E}_s \cup \mathcal{E}_u\), where \(\mathcal{E}_s = \{(I_3,0_{6\times1})\}\) and \(\mathcal{E}_u = \left\{(U\Lambda U^\top, 0)\,\middle|\, \Lambda = \mathrm{diag}(1, -1, -1),\, U \in \mathrm{SO}(3) \right\}.\) Moreover, the set \(\mathcal{E}_u\) is unstable, and the singleton \(\mathcal{E}_s\) is almost-globally asymptotically stable (AGAS).
\end{theorem}

\begin{proof}
The proof of Theorem~\ref{Theorem1} follows the same steps as the proof of Theorem~1 in \cite{tchonkeu2025barometer} and is therefore omitted here.
\end{proof}
Theorem~\ref{Theorem1} establishes almost global asymptotic stability (AGAS) of the estimation errors whenever the persistence of excitation condition in Lemma~\ref{Lemma1} is satisfied, overcoming the local convergence limitations reported for Pitot-based sensing in~\cite{oliveira2024pitot} while remaining consistent with the topological constraints of $\mathrm{SO}(3)$.

\subsection{Discrete-Time Implementation}
The proposed observer is implemented at the IMU sampling period~$T_s$. 
Over each interval $[t_k,t_{k+1})$, we assume the measured acceleration 
$\mathbf a_k$ and angular velocity $\boldsymbol\omega_k$ are constant 
(see~\cite{Bryne2017}). Let $\Omega_k \doteq \boldsymbol\omega_k^\times$ 
and $\theta_k \doteq \|\boldsymbol\omega_k\|T_s$. Approximating the discrete process noise by $S_{d,k}\approx S_k T_s$, the transition block associated with vectors expressed in the rotating body frame satisfies 
$\dot \phi_{11}(t)=-\Omega_k \phi_{11}(t)$ with $\phi_{11}(0)=I_3$, 
which integrates to the incremental rotation via Rodrigues' formula~\cite{Ma2004rodriguesformula}
\begin{equation}
\phi_{11,k} 
= I_3 - \frac{\sin\theta_k}{\|\boldsymbol\omega_k\|}\,\Omega_k
+ \frac{1-\cos\theta_k}{\|\boldsymbol\omega_k\|^2}\,\Omega_k^2 .
\label{eq:Phi22}
\end{equation}
Using this result, a first-order discretization of the continuous-time 
state and input matrices in~\eqref{eq:state_command_matrix} yields
\begin{equation}
A_{d,k} \approx
\begin{bmatrix}
\phi_{11,k} & gT_sI_3\\
\mathbf{0}_{3 \times 3}& \phi_{11,k}
\end{bmatrix},\qquad
B_{u_d,k} =
\begin{bmatrix}
T_sI_3 \\
\mathbf 0_{3\times 3}
\end{bmatrix}.
\label{eq:AdBd}
\end{equation}
The resulting discrete-time observer, summarized in 
Algorithm~\ref{algo:Discrete_proposed_obs}, follows the standard 
correction--prediction structure with the above state and input matrices. The attitude observer~\eqref{eq:attitude_observer} is discretized at the 
IMU frequency using exponential-Euler integration on $\mathrm{SO}(3)$~\cite{mahony2008nonlinear} (see line~23 of Algorithm~\ref{algo:Discrete_proposed_obs}). 

\begin{algorithm}[!t]
\caption{Discrete-Time Implementation of the proposed observer}
\label{algo:Discrete_proposed_obs}
\begin{algorithmic}[1]
\Statex \textbf{Input:} $\hat x_{0|0},\,P_{0|0},\hat R_0$; $\mathbf a_k,\boldsymbol{\omega}_k$ ; $\mathbf m_{\mathcal B,k}, \mathbf m_{\mathcal I,k}$; $y_{p,k}$; $T_s$; $S_k$; $Q_k$; gains $k_z,k_m$; Pitot geometry $B$.
\Statex \textbf{Output:} $\hat x_k,\,\hat R_{k},\,$ \(\forall k \in \mathbb{N}\)
\For{each time \(k \geq 0\)}
\If{IMU data \(\mathbf a_k\), \(\omega_k\) is available}
    \State $\Omega_k \gets \boldsymbol{\omega}_k^\times$,\quad $\theta_k \gets \|\omega_k\|\,T_s$, \quad $S_{d,k}\gets S_k\,T_s$;
    \State {$\phi_{11,k} \gets $~\eqref{eq:Phi22};}
    \State $A_{d,k}, B_{u_d,k} \gets$~\eqref{eq:AdBd};
    \State $\hat x_{k+1|k} \gets A_{d,k}\hat x_{k|k} + B_{d,k} \mathbf a_k$;
    \State $P_{k+1|k} \gets A_{d,k}P_{k|k}A_{d,k}^\top + S_{d,k}$
\EndIf
\If {Pitot data is available}
    \State $C_k \gets$\eqref{eq:output_matrix}; 
    \State $K_{k+1} \gets P_{k+1|k} C_k^\top (C_kP_{k+1|k} C_k^\top + Q_{k})^{-1}$
    \State \(\hat x_{k+1|k+1} \gets \hat x_{k+1|k}+ K_{k+1}\left(y_{p,k+1}-C_k \hat x_{k+1|k}\right)\)
    \State \(P_{k+1|k+1} \gets (I_6-K_{k+1}C_k)P_{k+1|k}\)
\Else
    \State $\hat x_{k+1|k+1}\gets \hat x_{k+1|k}$; $P_{k+1|k+1}\gets P_{k+1|k}$.
\EndIf
\State $P_{k+1|k+1}\gets \frac{1}{2}(P_{k+1|k+1} + P_{k+1|k+1}^{\top})$
\If{Magnetometer data is available or k= 0}
    \State $\bar{\mathbf m}_{\mathcal I,k} \gets \bar\pi_{\mathbf e_3}\mathbf m_{\mathcal I}$;\quad $\bar{\mathbf m}_{\mathcal B,k} \gets \bar\pi_{\hat z_k}\mathbf m_{\mathcal B,k}$
    \State $\sigma_{R_{m,k}} \gets k_m(\bar{\mathbf m}_{\mathcal I,k} \times \hat R_k \bar{\mathbf m}_{\mathcal B,k})$
\EndIf
\State $\sigma_{R,k} \gets k_z(\mathbf e_3 \times \hat R_k \hat z_k) + \sigma_{R_{m,k}} $
\State $\hat R_{k+1} \gets \hat R_k \exp\!\big((\boldsymbol{\omega}_k - \hat R_k^\top \sigma_{R,k})^\times T_s\big)$
\EndFor
\end{algorithmic}
\end{algorithm}
\section{EXPERIMENTS}
This section presents experimental results obtained from a flight test conducted on a MakeFlyEasy Fighter VTOL UAV illustrated in figure~\ref{fig:veh_pitot_imu}. Although the proposed formulation assumes slowly time-varying wind conditions, the flight was performed in a windy environment. Consequently, these experiments provide a meaningful assessment of the robustness of the proposed observer under realistic operating conditions.
The performance of the proposed discrete-time cascade observer was evaluated using real flight data recorded by a \textit{Pixhawk} flight controller running the PX4 autopilot~\cite{meier2015px4}. The onboard sensor suite includes an IMU providing specific force and angular rate measurements \(\mathbf a\) and \(\boldsymbol{\omega}\) respectively, a three-axis magnetometer \( \mathbf{m}_{\mathcal{B}}\in \mathbb{S}^2\), and a forward-aligned Pitot tube, i.e. $m=1$ and $B = \mathbf{e}_1$ equipped with an \textit{SDP33} Sensirion differential pressure sensor~\cite{SensirionSDP33}.

To assess the impact of sideslip information on the observer performance, two experimental configurations are considered. In Configuration~1, coordinated flight is assumed, in which a zero-sideslip condition is introduced through the pseudo-measurement $V_{a,2}=0$ ($\beta =0 $), as encoded in the measurement model \eqref{eq:pitot_meas_conf_1}. 
In Configuration~2, no pseudo-measurement is introduced, and the Riccati observer relies solely on a forward-aligned Pitot tube measurement as defined in \eqref{eq:pitot_meas_conf_2}.

The autopilot records flight logs for each onboard sensor and provides estimates of the vehicle attitude \(\bar{R}\) and the body-frame air velocity \(\bar{\mathbf{V}}_{a}\) (see Fig.~\ref{fig:Real_Air_Velocity}), computed internally by an extended Kalman filter \cite{px4_ecl_ekf_tuning}. In order to avoid the influence of magnetometer disturbances on the estimation process and to isolate the contribution of pitot and inertial measurements, the attitude estimate \(\bar R\) provided by the autopilot is used solely to reconstruct an equivalent inertial magnetic field according to
$
\mathbf{m}_{\mathcal{I}} := \bar R\,\mathbf{m}_{\mathcal{B}}.
$
This reconstructed inertial magnetic field is treated as a known reference and is used to evaluate both reduced and full attitude estimation errors, without directly relying on magnetometer measurements in the estimator. The air velocity estimate \(\bar{\mathbf{V}}_{a}\) and tilt estimate $\bar{z}:=\bar{R}^{\top}\mathbf{e}_3$ are used as a reference for assessing the accuracy of the proposed air velocity \(\hat{\mathbf{V}}_{a}\) and tilt \(\hat{z}\) estimates, respectively . The reliability of \(\bar{\mathbf{V}}_{a}\) is verified by comparing its first component \(\bar{V}_{a,1}\) with the Pitot-tube measurement \(V_{a,1}\), showing very close agreement (see Fig.~\ref{fig:comparison_real_vs_measured_Air_Velocity}).

In this experiment, the estimator was initialized over a selected flight segment in which the persistent excitation requirement in \eqref{eq:PE_Condition} holds. This guarantees uniform observability, which is evaluated by calculating the condition number of
\begin{equation}
\bar M(t_k,t_{k+1}) := \int_{t}^{t + \delta}\bar R(s)BB^{\top}\bar R^{\top}(s)ds, 
\end{equation}
over the interval $[t_k,t_{k+1}]$ with $t_k =k\delta$, $k$ a positive integer, and $\delta = 2 [s]$. In the two Configurations, the initial conditions at $t_0 = 40[s]$ are such that 
$\bar{\mathbf{V}}_a(t_0) \approx \begin{bmatrix}20.8&-0.48&5.9 \end{bmatrix}^{\top}[\mathrm{m/s}],$
$\hat{\mathbf{V}}_a(t_0) = \begin{bmatrix}10&2&0.3 \end{bmatrix}^{\top}[\mathrm{m/s}]$
$\bar{z}(t_0) \approx \begin{bmatrix}0.0049&0.0125&1.0 \end{bmatrix}^{\top}$
$\hat z(t_0) = \hat R(t_0)^{\top}\mathbf{e}_3$, where initial orientation estimate $\hat R(t_0)$ corresponds to the initial Euler angles with yaw  \(\pi/6\), pitch \(-\pi/18\), and roll \(\pi/9\). The observer gains are selected as $k_z=2$ $k_m=1$, with $S = \mathrm{diag}(0.02,0.01,0.01,10^{-4}I_3)$, $P(t_0)=\mathrm{diag}(116.6,6.15,3.3,0.6I_3)$. 
For Configuration~1, the Pitot sensor covariance $Q = \mathrm{diag}(\sigma^2_{V_{a,1}},\sigma^2_{V_{a,2}}),$ where $\sigma^2_{V_{a,1}}= 10^{-3}[(\mathrm{m/s})^2]$ models the Pitot tube measurement noise on $V_{a,1}$, and $\sigma^2_{V_{a,2}}$ corresponds to the variance associated with the pseudo-measurement on $V_{a,2}$.
Although $V_{a,2}$ (or equivalently) has zero mean, atmospheric turbulence and lateral–directional dynamics induce fluctuations around this equilibrium. To account for these uncertainties, a large variance is assigned to $\sigma^2_{V_{a,2}}$; in practice it is chosen as $\sigma^2_{V_{a,2}} \approx 10\sigma^2_{V_{a,1}}.$
For Configuration~2, the covariance is reduced to $Q=\sigma^2_{V_{a,1}}$. 

The accelerometer and gyroscope measurements $\mathbf{a}(t)$ and $\boldsymbol{\omega}(t)$ are discretized at $250~[\mathrm{Hz}]$, while the Pitot and magnetometer measurements are available at $50~[\mathrm{Hz}]$ and $10~[\mathrm{Hz}]$, respectively.
Figures~\ref{fig:AirVelocity_Components}, Figures~\ref{fig:Tilt_Component_Error}, and~\ref{fig:Euler_Angles_Att_Error} illustrate the air velocity components and full air velocity estimation error (\(\|\bar V_a-\hat V_a\| \)), the tilt components and reduced attitude error (\(\left\| \bar{z} - \hat z \right\|\)), and the Euler angles and full attitude estimation error (\(\operatorname{trace}(I_{3} - \hat{R}\bar R^\top)\)), respectively. 
The air velocity estimation results show a clear distinction between the two configurations. Throughout the flight, the forward component $\hat {\mathbf{V}}_{a,1}$ converges to the onboard reference $\bar{\mathbf{V}}_{a,1}$ for both cases due to the direct constraint imposed by the Pitot measurement. However, significant differences arise in the lateral and vertical components. Configuration~2 exhibits large root-mean-square error (RMSE) in the lateral and vertical air velocity components, approximately $12 [\mathrm{m/s}]$ and $9 [\mathrm{m/s}]$, respectively. This behavior is consistent with the higher values of the condition number of $\bar{M} (\mathrm{cond}(\bar M))$(see Fig.~\ref{fig:PE_Condition}), indicating poor numerical conditioning and loss of observability of the velocity direction dynamics when excitation is limited. In constrast, Configuration~1, maintains substantially lower $\mathrm{cond}(\bar M)$ values, thereby preserving practical observability. Between $t \in [43s, 70s]$, where the pitch angle varies significantly, the air velocity estimate error in Configuration~1 decreases to $\approx 3 [\mathrm{m/s}]$, showing that improved conditioning of $\bar M$ leads to better estimation accuracy. Moreover, during the coordinated flight segments $t \in [40.6s, 43s]$ and $t \in [70s, 80s]$ where the pitch angle remains nearly constant (see Fig.~\ref{fig:Euler_Angles_Att_Error}), Configuration~1 maintains lower $\mathrm{cond}(\bar M)$ values, and air velocity estimate $\mathbf{\hat V}_a$ converges to its actual values as expected. The reduced condition number of $\bar M$ in Configuration~1 also contributes to improved estimation performance for both reduced and full attitude states as shown in Figures~\ref{fig:Tilt_Component_Error} and \ref{fig:Euler_Angles_Att_Error}, respectively. In coordinated flight, all estimation errors remain bounded and converge to zero, confirming the asymptotic convergence and robustness properties of the proposed observer. Overall the results demonstrate that the pseudo-measurement acts as an effective observability-restoring mechanism that improves the estimation reliability for both air velocity and full attitude, even under weakly excited flight conditions and aggressive maneuvers.

\begin{figure}[!t]
    \centering
    \includegraphics[width=0.75\linewidth]{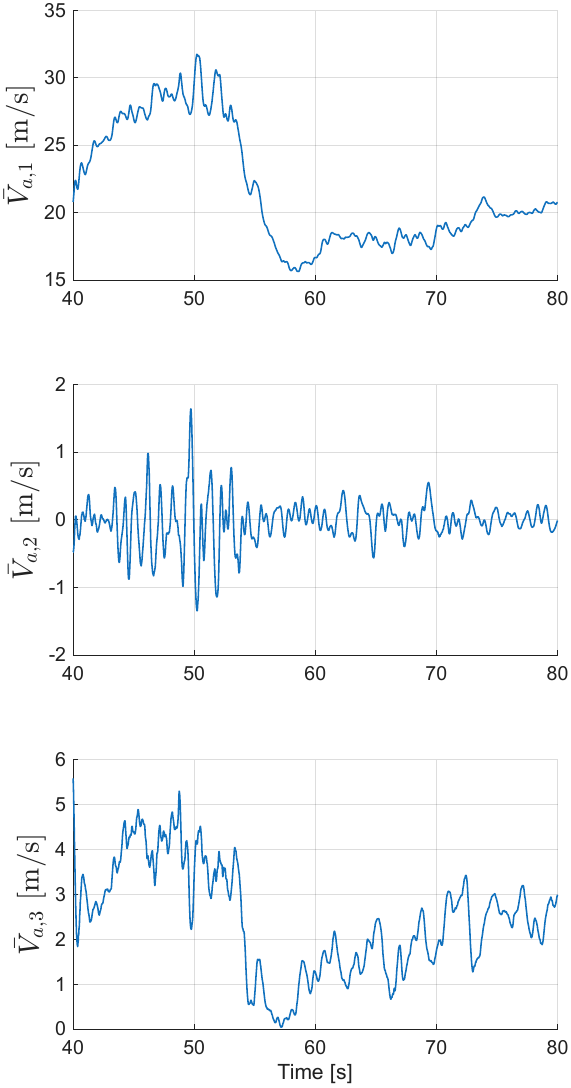}
    \caption{Onboard Air-velocity Components.}
    \label{fig:Real_Air_Velocity}
\end{figure}

\begin{figure}[!t]
    \centering
    \includegraphics[width=0.75\linewidth]{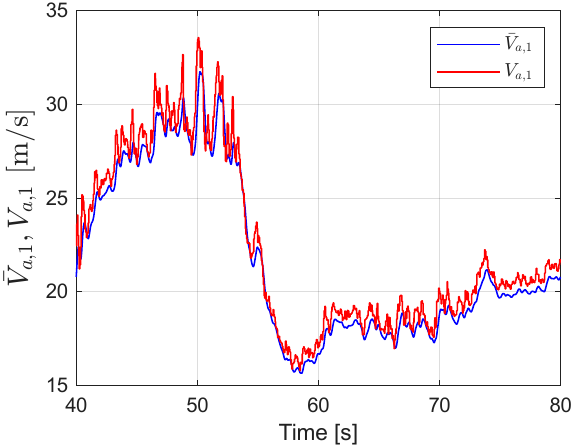}
    \caption{Comparison between measurement of $V_{a,1}$ (Red) and $\bar{V}_{a,1}$ (Blue) in the real flight data set.}
    \label{fig:comparison_real_vs_measured_Air_Velocity}
\end{figure}

\begin{figure}[!t]
    \centering
    \includegraphics[width=0.75\linewidth]{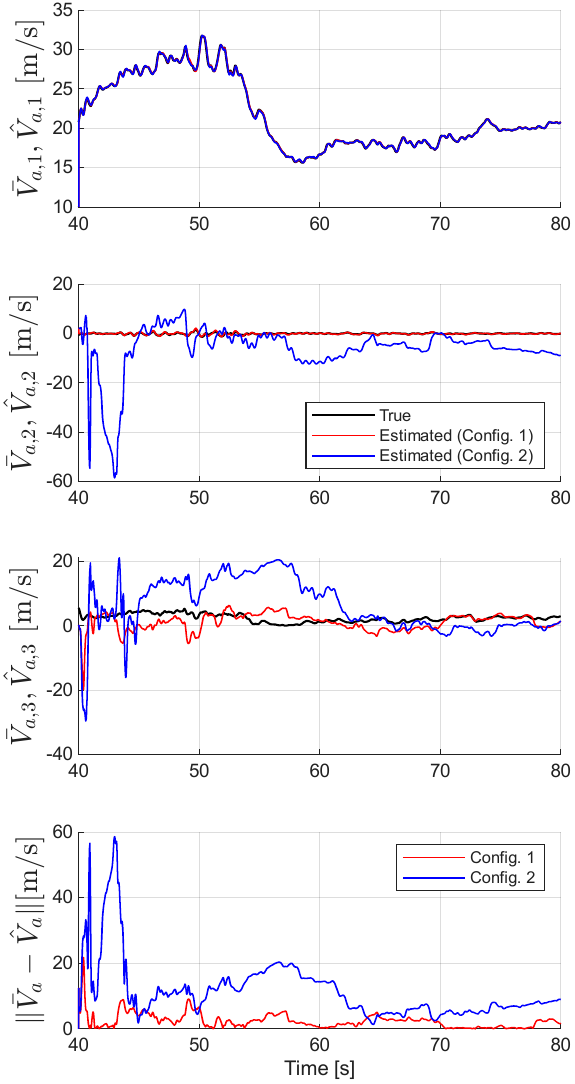}
    \caption{Air Velocity components and error: Configuration 1 (Red) and Configuration 2 (Blue).}
    \label{fig:AirVelocity_Components}
\end{figure}

\begin{figure}[!t]
    \centering
    \includegraphics[width=0.75\linewidth]{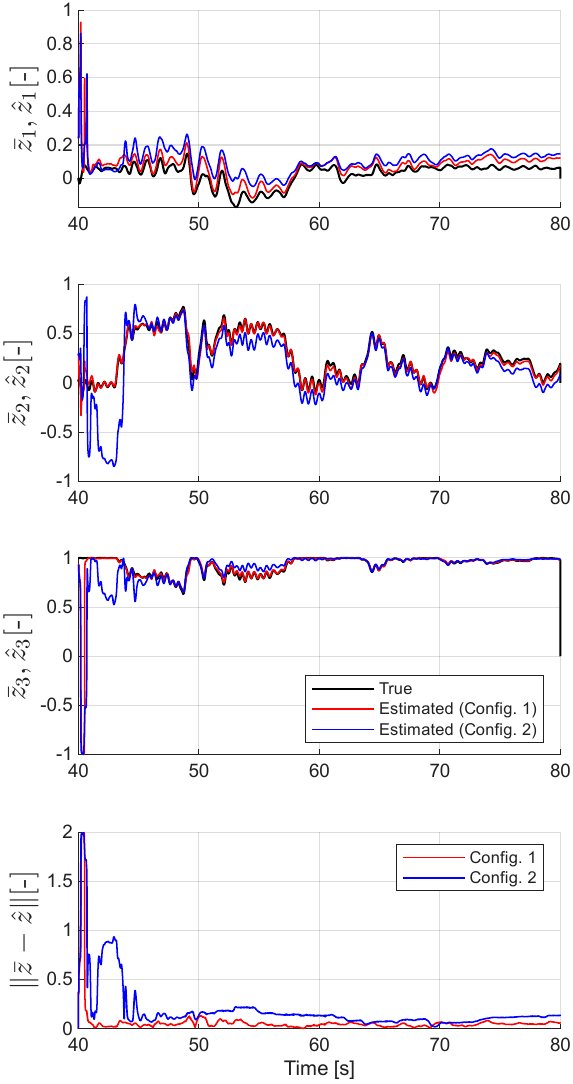}
    \caption{Tilt components and tilt error.}
    \label{fig:Tilt_Component_Error}
\end{figure}

\begin{figure}[!t]
    \centering
    \includegraphics[width=.75\linewidth]{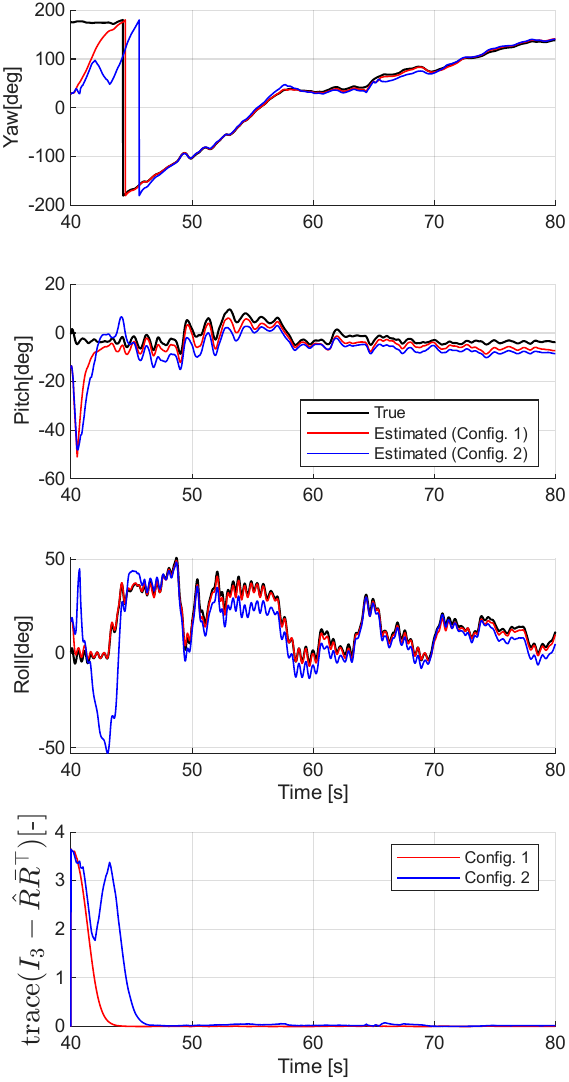}
    \caption{Euler angles and attitude error.}
    \label{fig:Euler_Angles_Att_Error}
\end{figure}

\begin{figure}[!t]
    \centering
    \includegraphics[width=.75\linewidth]{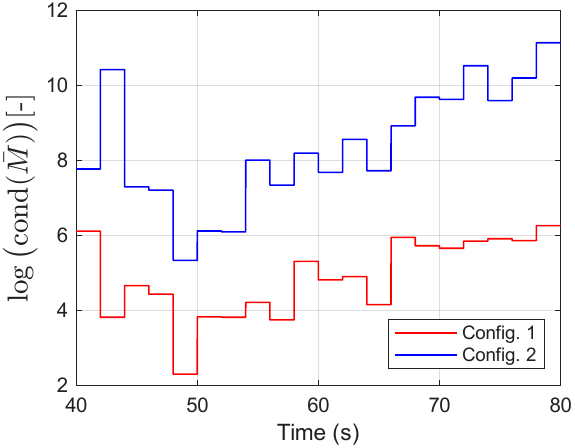}
    \caption{Evolution of cond$(\bar M)$ over the time.}
    \label{fig:PE_Condition}
\end{figure}

\section{CONCLUSION}
In this work, we developed a cascade observer for Pitot-aided estimation, combining a continuous-discrete Riccati observer for air velocity and gravity-direction dynamics with a nonlinear observer on $\mathrm{SO}(3)$ for full attitude estimation. Theoretical analysis (Lemma~\ref{Lemma1} and Theorem~\ref{Theorem1}) established almost-global asymptotic convergence of the interconnected error dynamics under persistently exciting motion, and experimental validation with flight data from a fixed-wing VTOL UAV confirmed convergence of both air velocity and attitude estimation errors. Incorporating a zero-sideslip constraint, consistent with coordinated-flight aerodynamics, was shown to enhance overall estimation performance. While the proposed observer demonstrated robustness on real flight data collected in a windy environment, the theoretical analysis assumes bounded, slowly time-varying wind and does not explicitly estimate sensor biases; extending the framework to jointly estimate wind components and inertial sensor biases constitutes a natural direction for future work.

\bibliographystyle{IEEEtran}
\bibliography{root}

\end{document}